\documentclass[12pt,reqno]{amsart}
\usepackage[final]{graphicx}
\usepackage{amsfonts}
\usepackage{amsmath}
\usepackage{amssymb}
\usepackage{amsthm}

 \usepackage[foot]{amsaddr}

\newtheorem{theorem}{Theorem}[section]
\newtheorem{lemma}[theorem]{Lemma}
\newtheorem{corollary}[theorem]{Corollary}

\newtheorem{remark}[theorem]{Remark}

\newcommand{\R}{{\mathord{\mathbb R}}}

\newcommand{\Sp}{{\mathop{\mathbb S}}}
\newcommand{\vau}{{\mathord{\mathbf v}}}
\newcommand{\weh}{{\mathord{\mathbf w}}}
\newcommand{\de}{{\mathord{\mathrm{d}}}}

\begin{document}
\title{Decay of Information for the Kac Evolution}

\author{F. Bonetto$^1$} 
\address{$1$ School of Mathematics, Georgia Institute of Technology Atlanta, GA 30332, United States of America}
\email{{bonetto@math.gatech.edu}}
\author{R. Han$^2$}
\address{$2$ School of Mathematics, Georgia Institute of Technology Atlanta, GA 30332, United States of America. Current address: Department of Mathematics, Louisiana State University, Baton Rouge, LA 70803, United States of America}
\email{{rhan@lsu.edu}}
\author{M. Loss$^1$}
\email{{loss@math.gatech.edu}}


\thanks{\copyright~2017 by the authors. Reproduction of this article by any means permitted for non-commercial purposes. U.S. National Science Foundation grants DMS-1907643 (F.B),
DMS-2053285 (R.H.) and DMS-1856645 (M.L) are gratefully acknowledged.}

\begin{abstract} 
We consider a system of $M$ particles in contact with a heat reservoir of $N\gg M$ particles. The evolution in the system and the reservoir, together with their interaction, are modeled via the Kac's Master Equation. We chose the initial distribution with total energy $N+M$ and show that if the reservoir is initially in equilibrium, that is if the initial distribution depends only on the energy of the particle in the reservoir, then the entropy of the system decay exponentially to a very small value. We base our proof on a similar property for the Information. A similar argument allows us to greatly simplify the proof of the main result in \cite{BGLR}. 
\end{abstract}

\maketitle

\section{Introduction}

In 1955 Mark Kac \cite{kac} introduced a simple model to study the evolution of a gas of $N$ particles undergoing pairwise collisions. Instead of following the deterministic evolution of the particle till a collision takes place, it is assumed that collisions happen randomly with every particle undergoing, on average, a given number of collision per unit time. Moreover, when a collision takes place, the energy of the two particles is randomly redistributed between them. In such a situation one can neglect the position of the particles and focus on their velocities. Finally, to obtain a model as simple as possible one assume that the particles move in one space dimension. This naturally leads to an evolution governed by a master equation for the probability distribution $F(\mathbf v)$, where $\mathbf v\in \mathbb{R}^N$ describes the velocities of the particles.

The Kac's master equation has proved to be very useful for understanding foundational issues of kinetic theory. In this context, Kac invented propagation of chaos and gave a satisfactory derivation of the spatially homogeneous Boltzmann-Kac equation.  All this is explained in great detail in Kac's original work \cite{kac} and \cite{KacBook} and will not be repeated here.

Kac's master equation also provides a natural setting to study approach to equilibrium. 

Exponential convergence to equilibrium in the sense of $L^2$ distance with a rate independent of $N$ was conjectured by Kac and it was established in \cite{Jeanvresse} while this rate was explicitly computed in \cite{CCL1}. A more natural way to define approach to equilibrium is via the entropy. This provides a better setting since the entropy, in general, grows only linearly with the number of particles. There is no result of exponential decay of entropy with a rate that is uniform in $N$ for the original Kac model. Moreover estimates of the entropy production rate seems to point to a slow decay of the entropy, at least for short times, see \cite{villani, amit}.

In a recent papers \cite{BLTV,BGLR} a different approach is proposed. One considers a {\it small} system with $M$ particles in contact with a {\it large} heat reservoir with $N$ particles. The probability distributions of the system plus reservoir is thus $F(\mathbf v, \mathbf w)$ where $\mathbf v \in \R^M, \mathbf w \in \R^N$.  Here  $\mathbf v = (v_1, \dots, v_M)$ describes the particles of the system and $\mathbf w = (w_{M+1}, \dots, w_{N+M})$ the particles of the reservoir. The evolution is given by the Kac's master equation
\begin{equation} \label{eq: masterequation}
\frac{\partial F}{\partial t} = \mathcal L F \ ,  \qquad F(\mathbf v, \mathbf w, 0) = 
F_0(\mathbf v, \mathbf w)  \ , 
\end{equation}
where 
\begin{equation}\label{eq:model}
\mathcal{L} = \frac{\lambda_S}{M-1} \sum_{1 \leq i<j \leq M} \left( R_{ij} - I 
\right) + \frac{\lambda_R}{N-1} \sum_{M < i<j \leq N + M} \left( R_{ij} - I 
\right) + \frac{\mu}{N} \sum_{i=1}^{M} \sum_{j=M+1}^{M+N} \left( R_{ij} - I 
\right) \ , 
\end{equation}
and $R_{ij}$ is given as follows. For $1 \le i<j \le M$ one has
\begin{align} \label{rotation}
(R_{ij}F)(\mathbf v, \mathbf w) = \int_{-\pi}^\pi  \nu(\de \theta) \, \de 
\theta \, F(r_{ij}(\theta)^{-1}(\mathbf v, \mathbf w)) \ , 
\end{align}
where
\begin{align}\label{eq:R}
r_{ij}(\theta)^{-1}(\mathbf v, \mathbf w) = (v_1, \dots, v_i \cos \theta - v_j 
\sin \theta, \dots, v_i \sin \theta  + v_j \cos\theta , \dots, v_M, \mathbf w) 
\ . 
\end{align}
The other $R_{ij}$s are defined analogously. The $R_{i,j}$ represents the effect of a collision between particle $i$ and particle $j$. The particles move in one dimension and upon collision energy is randomly redistributed between the two colliding particle. The probability distribution $\nu(\de \theta)$ plays the role of the scattering cross section. More assumptions will be made below.

The various constants have a simple interpretation in that $\lambda_S$ is the rate at which one particle from the system will scatter with any other particle in the system and similarly for $\lambda_R$. Likewise, $\mu$ is the rate at which a single particle of the system will scatter with any particle in the reservoir. 

The main point in \cite{BLTV,BGLR} is an analysis of the evolution of  a state where the reservoir is initially in equilibrium but the system is not. In this context equilibrium means that the initial distribution for the reservoir is given by a Maxwellian (Gaussian) function:
\begin{equation} \label{initialcondition}
F_0(\mathbf v, \mathbf w) = f_0(\mathbf v) e^{-\pi |\mathbf w|^2} 
\end{equation}
where we chose units in such a way that the inverse temperature is $\beta =2\pi$.  As time progresses the system and the reservoir interact and their joint state relax to a global equilibrium. On the other hand, the rate at which a particular particle from the reservoir will scatter with a particle in the system is given by $\mu M/N$. Hence, when $N$ is large compared to $M$ this process is suppressed and one expects that the reservoir does not move far from its equilibrium. This is indeed what was proved in \cite{BLTV} where the reader will find additional information on the physical ideas behind this model.

\begin{remark}\label{rem:local}
 As observed in \cite{BLTV,BGLR} the evolution in \eqref{eq: masterequation}, \eqref{eq:model} lends itself to a different interpretation. If we take 
 \[
  \lambda_S=\frac{2\lambda (M-1)}{N+M-1},\qquad \lambda_R=\frac{2\lambda (N-1)}{N+M-1},\qquad  \mu=\frac{2\lambda N}{N+M-1}
 \]
then we have that
\[
\mathcal L  =  \lambda(M+N)(Q-I)\qquad Q = \frac{1}{\left(\begin{array}{c} M+N \\2 \end{array}\right)} \sum_{1= i<j } ^{M+N}R_{i,j} \, .
\]
In this setting, together with the initial condition \eqref{initialcondition}, \eqref{eq: masterequation} represents the evolution of a Kac system with $N+M$ particle where, initially, $N$ particle are in equilibrium while $M$ are out of equilibrium, that is a ``local perturbation''.

\end{remark}

The central result in \cite{BGLR} concerns the entropy of the marginal on the system of the full distribution defined as:
\begin{equation}\label{marginal_R}
f(\mathbf v, t) =[\mathcal M e^{\mathcal{L}t}F_0](\mathbf v, t):= \int_{\R^N} \left[ e^{\mathcal L t} F_0\right](\mathbf v, \mathbf w) e^{-\pi |\mathbf v|^2}  \de \mathbf w \ ,
\end{equation}
Its entropy with respect to the Maxwellian distribution $\gamma(\mathbf v)=e^{\pi |\mathbf v|^2}$ is 
given by
\begin{equation*}
\widetilde S(f(\cdot, t)) := \int_{\R^M} f(\mathbf v, t) \log\left( \frac{f(\mathbf v, 
t)}{e^{-\pi |\mathbf v|^2}} \right) \,  \de \vau  \ .
\end{equation*}
Note that the entropy can be written in another way, which is easier to analyze. If we set
\begin{align}\label{def:h0}
f_0(\mathbf v) = e^{-\pi |\mathbf v|^2} h_0(\mathbf v) \ ,
\end{align}
then because any rotationally invariant distribution is invariant under the Kac's time evolution, we get
\[
e^{\mathcal L t} F_0 (\mathbf v, \mathbf w) = e^{-\pi(|\mathbf v|^2+ |\mathbf w|^2)}  \left[e^{\mathcal L t} h_0\right](\mathbf v,\mathbf w)
\]
It follows that we can write $f(\mathbf v, t) = e^{-\pi |\mathbf v|^2} h(\mathbf v, t)$ where
\[
h(\mathbf v, t) := \int_{\R^N} \left[ e^{\mathcal L t} h_0\right](\mathbf v, \mathbf w)  e^{-\pi |\mathbf w|^2} \de \mathbf w
\]
and, thus,
\begin{align}\label{def:S}
S(h(\cdot, t)): = \int_{\R^M} h(\mathbf v, t) \log (h(\mathbf v, 
t)) e^{-\pi |\mathbf v|^2}  \, \de \mathbf v= \widetilde S(f(\cdot, t)) \ .
\end{align}
%

The following theorem was proved in \cite{BGLR}:
\begin{theorem} \label{thm: main} Let $\nu(\de \theta) = \rho(\theta) \de \theta$ be a probability distribution with an absolutely convergent 
Fourier series such that
\begin{align} \label{assumprho}
\int_{-\pi}^\pi \nu(\de \theta)\, \sin \theta \cos \theta =0 
\  .
\end{align}
The entropy of $f(\mathbf v, t)$ {\it relative} of to the thermal state $e^{-\pi |\mathbf 
v|^2}$ then satisfies  
\begin{equation*}
\tilde{S}(f(\cdot, t)) \le \left[\frac{M}{N+M} + \frac{N}{N+M} e^{- t \mu_{\rho} 
(N+M)/N} \right] \tilde{S}(f_0)  \ , 
\end{equation*}
where
\begin{equation*} 
\mu_{\rho} = \mu \int_{-\pi}^\pi  \rho(\theta) \sin^2(\theta)\, \de \theta \, 
 \ , 
\end{equation*}
and $f_0$ is the initial condition of the system  introduced in \eqref{initialcondition}.
\end{theorem}

\begin{remark}
 In \cite{BGLR} the extra condition that $N>M$ is added to Theorem \ref{thm: main}.  The proof of the theorem presented below does not require any condition on $N$ and $M$ and on closer inspection this condition is also unnecessary for the proof presented in \cite{BGLR}.
\end{remark}
\begin{remark}
Condition \eqref{assumprho} follows from the natural assumption of microscopic reversibility, that is $\rho(\theta)=\rho(-\theta)$. In our setting, it is only needed in the proof of the summation rule in Lemma \ref{sumrule}.
\end{remark}

The point about Theorem \ref{thm: main} is that, when the size of the reservoir $N$ is large compared to the size of the system $M$, the entropy of the marginal converges exponentially fast to a small value with a rate that
is essentially independent of the size of the system. Observe, however, that the limiting state $\tilde f(\mathbf v)=\lim_{t\to\infty}f(\mathbf v,t)$ is, in general, not a Maxwellian distribution so that $S(\tilde f)>0$. The assumption on the initial condition of the system is very weak and, as explained in \cite{BGLR}, the theorem is essentially sharp.

There are, however, some unsatisfactory aspects of this result or rather its proof. For one, the assumption on the smoothness of the collision probabilities is rather complicated and is there for technical reasons. More importantly, while the proof has a natural ring to it, it is rather involved. More serious is that the method does not seem to carry over to the standard Kac model which operates on the sphere $\Sp^{M+N-1}(\sqrt{N+M})$. The proof uses the expansion in terms of collision histories and estimates the contribution to the entropy by  a fixed number of collisions using Nelson's hyper-contractive estimate and then estimates the correlations using Brascamp-Lieb type inequalities together with a re-summation of the series. It is very unclear how to carry such a program over to the sphere case. Despite the wonderful analogy between large dimensional spheres and Gauss space the connection is not uniform enough that Theorem \ref{thm: main} carries over to the sphere case.

The point of this paper is to remedy the faults described above. A very natural way to think about such problem is to think in terms of the information. There is a well known connection between entropy and  information which was used by Ledoux \cite{Ledoux} to give a very simple proof of Nelson's sharp hyper-contractive estimate. Not only does this argument yield the inequality, but it delivers all the functions for which there is equality.

The information of a function $f>0$, normalized with respect to Gauss measure $e^{-\pi |\mathbf{v}|^2} \de\mathbf v$ is given by
\begin{align}\label{def:Information}
I(f) = \int_{\R^M} \frac{|\nabla f(\mathbf{v})|^2}{f(\mathbf{v})}e^{-\pi |\mathbf{v}|^2} \de  \mathbf{v} \ .
\end{align}
Let $P_s$ be the Ornstein-Uhlenbeck semigroup on $\R^d$ defined by
\begin{align}\label{def:Ps}
[P_s f](\mathbf x) = \int_{\R^d }f(e^{-s}\mathbf x+ \sqrt{1-e^{-2s}} \mathbf y) e^{-\pi \mathbf |y|^2}   \de \mathbf y.
\end{align}
In our context, $d$ could be equal to $M$ or $M+N$ depending on the function space it acts on.

An important relation between the entropy and information is given by the equation
\begin{equation} \label{entropyinfo}
S(f)=\int_{\R^M} f  \log (f) e^{-\pi|\mathbf v|^2} \de\mathbf{v} = \int_0^\infty I(P_sf) \de s \ .
\end{equation}
The connection between information and entropy for the various Kac models that we study will be explored further in Section \ref{sec:Information}, which will allow us to pass from decay of Information to that of Entropy.

Our first result is a version of the previously  mentioned theorem in terms of the information. 
\begin{theorem} \label{main1}
Assume that $\nu$ is a probability measure on the circle subject to the condition
\begin{align} \label{assumnu}
\int_{-\pi}^\pi \nu(\de \theta) \, \sin \theta \cos \theta =0 \ .
\end{align}
Then
 \begin{equation*}
I(h(\cdot, t)) \le  \left[\frac{M}{N+M} + \frac{N}{N+M} e^{- t \mu_{\nu} 
(N+M)/N} \right]I(h_0)  \ , 
\end{equation*}
where
\begin{equation*} 
\mu_{\rho} = \mu \int_{-\pi}^\pi  \nu ( \de \theta)  \, 
\sin^2(\theta) \ , 
\end{equation*}
and $h_0$ is as introduced in \eqref{initialcondition} and \eqref{def:h0}.
\end{theorem} \label{main2}
As a corollary we have
\begin{corollary}
Theorem \ref{thm: main} continuous to hold under the assumption that the probability measure $\nu$  satisfies the condition
\begin{align}
\int_{-\pi}^\pi \nu(\de \theta) \, \sin \theta \cos \theta =0 
\  .
\end{align}
\end{corollary}

We turn now our attention to the original setting of Kac's model as an evolution on $\Sp^{M+N-1}(\sqrt{M+N})$.
The master equation is the same as \eqref{eq: masterequation} and \eqref{eq:model} but now the distribution $F$ is a function in $L^1(\Sp^{M+N-1}(\sqrt{M+N}))$ where the measure is taken to be the normalized uniform measure on the sphere. It is obvious that the constant function $1$ is an equilibrium state. Under additional assumption on $\nu(\de \theta)$ one can show that it is the unique equilibrium state, however, this issue will not be considered here.

We implement the idea of system versus reservoir in the following fashion. We assume that the initial condition $F_0(\mathbf v, \mathbf w)$ is a function that is invariant under all rotations that fix
the variables $\mathbf v$, i.e., $F_0(\mathbf v, \mathbf w)$ depends on $\weh$ only through $|\weh|$.

Given an arbitrary distribution $F(\mathbf v, \mathbf w)$ we define its marginal to be the average of $F$ over all rotation that fix the vector $\mathbf v$. In other words denoting the operation of taking the marginal with respect to the $\mathbf w$ variables by $\mathcal N$, we have
\begin{equation}\label{def:marginal}
\mathcal N F(\mathbf v, \mathbf w) = \int_{SO(N)} F(\mathbf v, R^{-1} \mathbf w) \de R
\end{equation}
where we integrate over the normalized Haar measure on $SO(N)$. Hence $\mathcal N F$ depends on
$\vau$ and $|\weh|$ only. As before we consider the Entropy
\[
S(F) = \int_{\Sp^{M+N-1}(\sqrt{M+N})}  F \log( F) \de \sigma
\] 
where $\sigma$ denotes the normalized volume measure on $\Sp^{M+N-1}(\sqrt{M+N})$. 
Likewise, the information of a state $F$ is given by
\begin{align}\label{def:Inf_sphere}
I(F)=\int_{\Sp^{M+N-1}(\sqrt{M+N})} \frac{|\nabla F|^2}{F}  \de \sigma
\end{align}
where the length $|\nabla F|^2$ is taken with respect to metric on the sphere.
In this case, our theorems are similar to the ones mentioned before but not quite as sharp.
\begin{theorem} \label{spherical}
As before let $\nu$ be a measure on the circle that satisfies
\[
\int_\Sp \nu(\de \theta) \sin \theta \cos \theta = 0 \ .
\]
Let $F_0(\mathbf v, \mathbf w), \mathbf v \in \R^M, \mathbf w \in \R^N$ be a distribution in $L^1(\Sp^{M+N-1}(\sqrt{M+N}))$ invariant under all rotations that fix the variable $\mathbf v$. Then 
\[
I(\mathcal N e^{\mathcal Lt} F_0) \le \left[ 2\left(\frac{M}{M+N} + e^{-\mu t} \frac{N}{M+N} \right) +
\frac MN\right]  I(F_0) \ . 
\]
\end{theorem} 

As a corollary we get

\begin{corollary} \label{sphericalentropy}
With the assumption as in Theorem \ref{spherical} we have that
\[
S(\mathcal N e^{\mathcal L t} F_0) \le  \left[ 2\left(\frac{M}{M+N} + e^{-\mu t} \frac{N}{M+N} \right) +
\frac MN\right] S(F_0) \ .
\]
\end{corollary}

One notes that the result does not appear to be optimal in various respects. For one the rate is a bit worse than the one in the Gaussian case. Moreover, the estimate is not sharp at time $t=0$. What is more interesting is that as $t$ tends to infinity the various quantities do not go to zero. We know, however, that the information as well as the entropy do tend to zero as time tends to infinity. As noted after Theorem \ref{thm: main}, in the Maxwellian case the asymptotic state of the system is not a Maxwellian distribution so that the information and entropy need not vanish. In the spherical case, however, the constant function is the unique equilibrium state and thus the asymptotic state of the system.  We suspect that although the entropy becomes small very fast, it may take a very long time to bring it to zero. In other words, we suspect that an estimate of the type 
\[
S(\mathcal N e^{\mathcal L t} F_0) \le  e^{-\mu t}  S(F_0)
\]
might be wrong for large times. During the evolution the full state, i.e, the reservoir together with the system get correlated and our estimate provides a bound on how large this correlation can get. It might, however take a long time for the evolution to bring this correlation to zero. This is, of course, pure speculation and it would be nice to find a way to complete the picture.

Observe that a distribution as in \eqref{initialcondition}, restricted to the sphere $\Sp^{M+N-1}(\sqrt{M+N})$ delivers a distribution that depends on $\weh$ only through $|\weh|$ while the converse is clearly not true. This seems to suggest that it may be possible to extend Theorem \ref{thm: main} to a larger set of initial condition, at the cost of less optimal estimates, using Corollary \ref{sphericalentropy}. Due to the differences in the definition of marginal in \eqref{marginal_R} and \eqref{def:marginal} we were unable to exploit this connections.

Finally, although the most interesting case is when $N$ is much larger then $M$, we observe that for $N=1$ and with the rates chosen as in Remark \ref{rem:local}, Theorem \ref{thm: main} tell us that the entropy decreases form $S(h_0)$ to $(1-\frac1{M+1})S(h_0)$ with an exponential rate 1. This looks consistent with the estimates of slow entropy decrease in \cite{amit, villani}.

The paper is organized as follows. In Section 2 we give a proof of Theorem \ref{main1} and in Section 3 we review the connection between entropy and information which leads to the proof of Corollary \ref{main2}.
In Section 4 we present the proof of Theorem \ref{spherical} which together with the remarks on the connection of entropy with the information yields Corollary \ref{sphericalentropy}.

\section{Proof of Theorem \ref{main1}}
It will be convenient to write the generators $\mathcal L$ in the form
\[
\mathcal L = \Lambda (Q -I)
\]
where
\[
\Lambda = \frac12(\lambda_SM+\lambda_RN +2\mu M)
\]
and
\[
Q = \frac{\lambda_S}{(M-1)\Lambda} \sum_{1 \leq i<j \leq M} R_{ij}+ \frac{\lambda_R}{(N-1)\Lambda} \sum_{M < i<j \leq N + M} R_{ij} + \frac{\mu}{N\Lambda} \sum_{i=1}^{M} \sum_{j=M+1}^{M+N} R_{ij}  \ , 
\]
which is a convex combination of rotational averages.
The time evolution is is then given by
\[
e^{\mathcal L t} = e^{-\Lambda t} \sum_{k=0}^\infty \frac{\Lambda^k t^k}{k!} Q^k
\]
and 
\[
Q^k = \sum_{\alpha_1, \dots, \alpha_k} \lambda_{\alpha_1} \cdots \lambda_{\alpha_k} R_{\alpha_1} \cdots R_{\alpha_k}
\]
where the $\alpha_j$ range over all possible choices of colliding particles and $\lambda_{\alpha_j}$
is a weight that has one of the three values $ \frac{\lambda_S}{(M-1)\Lambda},\frac{\lambda_R}{(N-1)\Lambda}$ and $   \frac{\mu}{N\Lambda} $. Note that $\sum_\alpha \lambda_\alpha =1$. Recall that
\[
R_{i,j} h_0 (\mathbf v, \mathbf w) = \int_0^{2\pi} h_0(r_{i,j}^{-1}(\theta) (\mathbf v, \mathbf w)) \nu(\de \theta)
\]
where $r_{i,j}(\theta)$ is the rotation in the plane specified by the pair $i,j$ by the angle $\theta$. Hence
\[
Q^kh_0(\mathbf v, \mathbf w) = \sum_{\alpha_1, \dots, \alpha_k} \lambda_{\alpha_1} \cdots \lambda_{\alpha_k} \int_0^{2\pi} \cdots \int_0^{2\pi}  h_0([\Pi_{j=1}^k r_{\alpha_j}(\theta_{\alpha_j})]^{-1} (\mathbf v, \mathbf w)) \nu(\de \theta_{\alpha_1}) \cdots \nu(\de \theta_{\alpha_k}) 
\]
Next we write the $(M+ N) \times (M+N)$ rotation as
\[
[\Pi_{j=1}^k r_{\alpha_j}(\theta_{\alpha_j})] = \left[ \begin{array}{cc} A(\underline \alpha, \underline \theta)  & B(\underline \alpha, \underline \theta) \\ C(\underline \alpha, \underline \theta) & D(\underline \alpha, \underline \theta) \end{array}\right]
\]
where $A$ is $M \times M$, $D$ is $N \times N$, $B$ is $M \times N$ and $C$ is $N\times M$. Here
$\underline \alpha$ denotes the $k$-tuple of pairs and likewise $\underline \theta$ denotes the $k$-tuple of angles. Next we consider the function
\[
 h_{\underline \alpha, \underline \theta}(\mathbf v, \mathbf w)=h_0(A(\underline \alpha, \underline \theta) ^T \mathbf v+ C(\underline \alpha, \underline \theta)^T \mathbf w)
\]
and take its marginal
\[
\mathcal{M}[h_{\underline \alpha, \underline \theta}](\mathbf v)=\int_{\R^N} h_0(A(\underline \alpha, \underline \theta) ^T \mathbf v+ C(\underline \alpha, \underline \theta)^T \mathbf w) e^{-\pi|\mathbf w|^2} \de \mathbf w
\]
and compute the Information
\[
I(\mathcal{M}[h_{\underline \alpha, \underline \theta}])=\int_{\R^M} \frac{|\int_{\R^N} A(\underline \alpha, \underline \theta) \nabla h_0(A(\underline \alpha, \underline \theta) ^T \mathbf v+ C(\underline \alpha, \underline \theta)^T \mathbf w) e^{-\pi|\mathbf w|^2} \de \mathbf w|^2
}{\int_{\R^N} h_0(A(\underline \alpha, \underline \theta) ^T \mathbf v+ C(\underline \alpha, \underline \theta)^T \mathbf w) e^{-\pi|\mathbf w|^2} \de \mathbf w
} e^{-\pi |\mathbf v|^2} \de \mathbf v
\]
which, using Schwarz's inequality, is bounded above by
\[
\int_{\R^M} \int_{\R^N} \frac{|A(\underline \alpha, \underline \theta) \nabla h_0(A(\underline \alpha, \underline \theta) ^T \mathbf v+ C(\underline \alpha, \underline \theta)^T \mathbf w) |^2
}{h_0(A(\underline \alpha, \underline \theta) ^T \mathbf v+ C(\underline \alpha, \underline \theta)^T \mathbf w)}
e^{-\pi|\mathbf w|^2} \de \mathbf w\, e^{-\pi |\mathbf v|^2} \de \mathbf v \ .
\]
Changing variables  
\[
\mathbf v' =A(\underline \alpha, \underline \theta) ^T \mathbf v+ C(\underline \alpha, \underline \theta)^T \mathbf w \ , \mathbf w' = B(\underline \alpha, \underline \theta) ^T\mathbf v  +  C(\underline \alpha, \underline \theta)^T \mathbf w
\]
yields the bound
\begin{align}\label{eq:1_main1}
I(\mathcal{M}[h_{\underline \alpha, \underline \theta}])\leq
\int_{\R^M} \frac{|A(\underline \alpha, \underline \theta) \nabla h_0(\mathbf v) |^2
}{h_0(\mathbf v)}
 e^{-\pi |\mathbf v|^2} \de \mathbf v.
\end{align}
Using the convexity of the Information and \eqref{eq:1_main1}, we find that
\begin{equation} \label{series}
I(h(\cdot, t)) \le e^{-\Lambda t}\sum_{k=0}^\infty \frac{\Lambda^kt^k}{k!} \int_{\R^M} \frac{\nabla h_0(\mathbf v)^T K \nabla h_0(\mathbf v)}{h_0(\mathbf v)} e^{-\pi |\mathbf v|^2} \de \mathbf v \ ,
\end{equation}
where 
\begin{align*}
K: = \sum_{\alpha_1, \dots, \alpha_k} \lambda_{\alpha_1} \cdots \lambda_{\alpha_k} \int_{[0,2\pi]^k} \nu(\de \theta_1) \cdots \nu(\de \theta_k) A(\underline \alpha, \underline \theta)^T A(\underline \alpha, \underline \theta).
\end{align*}

The following lemma was proved in \cite{BGLR}. For the readers' convenience, we include the proof in the appendix.
\begin{lemma} \label{sumrule}
Assume that $\int_0^{2\pi} \nu(\de \theta) \cos \theta \sin \theta =0$. Then $K=\mathcal{C}_k I_M$, where
where 
\begin{align}\label{def:Ck}
\mathcal{C}_k := \left[\frac{M}{N+M}+ \frac{N}{M+N}\left(1-\mu_\nu
\frac{M+N}{\Lambda N}\right)^k \right] \ .
\end{align}
\end{lemma}

Applying Lemma \ref{sumrule} to \eqref{series} clearly yields Theorem  \ref{main1}.\qed

\bigskip

\section{Connection between Information and Entropy}\label{sec:Information}
First we briefly mention some general connections between the Information and Entropy.
Recall the Ornstein-Uhlenbeck semigroup on $\R^d$ as defined in \eqref{def:Ps},
\begin{align*}
[P_s f](\mathbf x) = \int_{\R^d }f(e^{-s}\mathbf x+ \sqrt{1-e^{-2s}} \mathbf y) e^{-\pi \mathbf |y|^2}   \de \mathbf y,
\end{align*}
with its generator 
\[
L = \frac{1}{2\pi} \Delta - \mathbf{x}\cdot \nabla,
\]
that satisfies
\begin{equation}\label{genera}
\int_{\R^d} (\nabla f)\cdot  (\nabla g)\, e^{-\pi |\mathbf x|^2} \de \mathbf x  = -\int_{\R^d} f\, (L g)\, e^{-\pi |\mathbf x|^2} \de \mathbf x \ .
\end{equation}
Let $I(f)$ be the information of $f$ as defined in \eqref{def:Information}, 
\begin{align*}
I(f) = \int_{\R^M} \frac{|\nabla f(\mathbf{v})|^2}{f(\mathbf{v})}e^{-\pi |\mathbf{v}|^2} \de  \mathbf{v}.
\end{align*}
We have
\[
I(P_sf) = e^{-2s} \int_{\R^M}  \frac{| P_s( \nabla f)(\mathbf x)|^2}{P_s(f)(\mathbf x)} e^{-\pi |\mathbf x|^2} \de \mathbf x,
\]
here $P_s(\nabla f):=(P_s(\partial_1 f), P_s(\partial_2 f),...,P_s(\partial_M f))$.
Schwarz's inequality shows that
$$
I(P_sf) \le e^{-2s} \int_{\R^M} \int_{\R^M} \frac{|\nabla f(e^{-s} \mathbf x + \sqrt{1-e^{-2s}} \mathbf y)|^2}{f(e^{-s} \mathbf x + \sqrt{1-e^{-2s}} \mathbf y)}  e^{-\pi (|\mathbf x|^2+|\mathbf y|^2)} \de \mathbf x \de \mathbf y
$$
and a simple change of variables yields
\begin{equation} \label{contractiveinfo}
I(P_sf) \le e^{-2s} I(f) \ .
\end{equation}
Using \eqref{genera} we can link the information with the derivative of the entropy, that is
\[
\frac{d}{ds}S(P_sf)=\frac{d}{ds} \int_{\R^M} P_sf \log (P_s f )e^{-\pi |\mathbf x|^2} \de \mathbf x =  \int_{\R^M} L(P_s f)  \log (P_s f)e^{-\pi |\mathbf x|^2} \de \mathbf x 
= - I(P_s f) 
\]
that proves \eqref{entropyinfo}.

Combining \eqref{contractiveinfo} and \eqref{entropyinfo} we immediately get
\[
 S(P_sf)=\int_s^{\infty} I(P_tf)\de t \le e^{-2s}\int_0^{\infty} I(P_tf)\de t=e^{-2s} S(f) \
\]
Moreover, the logarithmic Sobolev inequality is now an easy consequence. Indeed, we obtain
$$
S(f)= \int_0^\infty  I(P_sf) \de s \le I(f) \int_0^\infty e^{-2s} \de s = \frac12 I(f) \ .
$$
\medskip

Now we further explore the connections between the Information and Entropy for the Kac's models, which show that estimates on the Information transfer to estimates on the Entropy.
For the Kac's model on $\R^{M+N}$ we have
\begin{lemma}\label{lem:Inf_to_Entropy_R} 
Assume that for any probability distribution $H(\mathbf v, \mathbf w)$ in $L^1(R^{M+N},e^{-\pi(|\mathbf v|^2 +|\mathbf w|^2)} \de \mathbf v\, \de \mathbf w)$ one has
\[
I(\mathcal M (e^{\mathcal L t} H)) \le C(t) I(\mathcal M(H)),
\]
where $\mathcal{M}(H)$ is the marginal of $H$ as defined in \eqref{marginal_R} and $C(t)$ does not depend on $H$. 
Then
\[
S(\mathcal M (e^{\mathcal L t} H)) \le C(t) S(\mathcal M(H)) \ .
\]
\end{lemma}

\begin{proof}
First, by \eqref{entropyinfo}, we have
\begin{align}\label{eq:SMLt}
S( \mathcal M e^{\mathcal L t} H) = \int_0^\infty I (P_s \mathcal M (e^{\mathcal L t} H)) \de s \ .
\end{align}
Observe that
\begin{align}\label{eq:LtPs_commute}
e^{\mathcal L t} P_s = P_s e^{\mathcal L t}.
\end{align}
This follows since the generator of the Kac evolution is an average over rotations and $P_s$ commutes  with rotations. 
Also, by a simple change of variables one easily sees that
\begin{align}\label{eq:PsM_commute}
P_s\mathcal M=\mathcal M P_s\ .
\end{align}
Applying \eqref{eq:LtPs_commute} and \eqref{eq:PsM_commute} in \eqref{eq:SMLt}, we have
\begin{align*}
S( \mathcal M e^{\mathcal L t} H)=\int_0^\infty I (P_s \mathcal M (e^{\mathcal L t} H)) \de s
&=\int_0^{\infty} I(\mathcal{M} e^{\mathcal{L}_t} (P_s H)) \de s\\
&\leq C(t) \int_0^{\infty} I(\mathcal{M}  (P_s H)) \de s=C(t) S(\mathcal{M} (H)).
\end{align*}
This proves the claimed result.
\end{proof}

Very similar is the situation with the spherical Kac model. Again, we exploit an important symmetry;
the generator of the Kac master equation being an average over rotations commutes with the spherical Laplacian
\begin{align}\label{eq:Deltas_L}
e^{\Delta s} e^{\mathcal{L}_t}=e^{\mathcal{L}_t} e^{\Delta s}.
\end{align} 
Let $H(\mathbf v, \mathbf w)$ be a probability distribution on $\Sp^{M+N-1}(\sqrt{M+N})$. In \eqref{def:marginal} we defined its marginal as the average of $H$ over all rotation that fix the vector $\mathbf v$. If $e^{\Delta s}$ is the heat semi-group on the sphere, then
\begin{align}\label{eq:Deltas_N}
e^{\Delta s} \mathcal N = \mathcal N e^{\Delta s}.
\end{align}
This follows from the fact that the spherical Laplacian and hence the heat semi-group commutes with rotations. 
We have the following analogy to \eqref{entropyinfo}.
\begin{lemma}
Let $H(\mathbf v, \mathbf w)$ be a probability distribution on $\Sp^{M+N-1}(\sqrt{M+N})$. 
Then
\[
S(\mathcal N (e^{\mathcal L t} H)) = \int_0^\infty I(e^{\Delta s} \mathcal N (e^{\mathcal L t} H)) \de s \ .
\]
\end{lemma}
Similar to Lemma \ref{lem:Inf_to_Entropy_R}, we have
\begin{lemma}\label{lem:Inf_to_Entropy_S}  
Assume that for any probability distribution
\[
I(\mathcal N (e^{\mathcal L t} H)) \le C(t) I(\mathcal N(H))
\]
where $C(t)$ does not depend on $H$. Then
\[
S(\mathcal N (e^{\mathcal L t} H)) \le C(t) S(\mathcal N(H)) \ .
\]
 \end{lemma}
The proof uses \eqref{eq:Deltas_L} and \eqref{eq:Deltas_N}. It is very similar to that of Lemma \ref{lem:Inf_to_Entropy_R} and we shall omit it here.

 \section{Proof of Theorem  \ref{spherical}}

Let $F(\mathbf v, \mathbf w)$ be a smooth function on the sphere $\Sp^{M+N-1}(\sqrt{M+N})$, or alternatively a function on $\R^{M+N}$ homogeneous of degree zero. The variable $\mathbf v \in \R^M$ and $\mathbf w \in \R^N$.  We denote by $L F$ its angular momentum, i.e., 
the angular momentum is for $1\le i < j \le M$
\begin{equation*}
(v_i \partial_{v_j} - v_j \partial_{v_i})F \ ,
\end{equation*}
for $1 \le i \le M < j \le M+N$ it is
\begin{equation*}
(v_i \partial_{w_j} - w_j \partial_{v_i})F \ ,
\end{equation*}
and for $M< i< j \le M+N$
\begin{equation*}
(w_i \partial_{w_j} - w_j \partial_{w_i})F \ .
\end{equation*}
We write
\begin{equation*}
|LF|^2=|L^{\rm int} F|^2+|L^{\rm out} F|^2,
\end{equation*}
where 
\begin{equation*}
\begin{cases}
|L^{\rm int}F|^2:=\sum_{1\le i<j\le M} |(v_i \partial_{v_j} - v_j \partial_{v_i})F|^2
+\sum_{1\le i\le M<j\le M+N} |(v_i \partial_{w_j} - w_j \partial_{v_i})F|^2 \ ,\\
|L^{\rm out}F|^2:=\sum_{M< i< j \le M+N} |(w_i \partial_{w_j} - w_j \partial_{w_i})F|^2.
\end{cases}
\end{equation*}
In particular, if a function $F$ is independent of $\mathbf w$, then 
\begin{equation}\label{angular_ind_w}
|LF|^2=|L^{\rm int}F|^2.
\end{equation}
The Information $I(F)$, as defined in \eqref{def:Inf_sphere}, written in terms of the angular momentum is 
\begin{align*}
I(F)=\int_{\Sp^{M+N-1}(\sqrt{M+N})}\frac{|LF|^2}{F} \de \sigma.
\end{align*}

In the following, let $R$ be the rotation in $SO(M+N)$ which we write as
$$
R = \left( \begin{array}{cc} A & B \\ C & D \end{array} \right) \ .
$$
Let $G$ be the average
$$
G(\vau ,\weh):=\int_{SO(N)} F(A^T\vau +C^T S^T\weh, B^T\vau+D^TS^T\weh) \de S=\int_{SO(N)} F\circ R^{-1}(\vau, S^T \weh) \de S  \ .
$$

First, we prove 
\begin{lemma}\label{lem:ave1}
\begin{align*}
I(G)\leq 
I(F)
-\int_{\Sp^{M+N-1}(\sqrt{M+N})} \frac{|L^{\rm out} (F\circ R^{-1})|^2}{F\circ R^{-1}} \de \sigma.
\end{align*}
\end{lemma}

\begin{proof}
Let us note that since $G$ is independent of $\weh$, we have by \eqref{angular_ind_w}, 
\begin{equation*}
I(G)=\int_{\Sp^{M+N-1}(\sqrt{M+N})}\frac{|LG|^2}{G} \de \sigma=\int_{\Sp^{M+N-1}(\sqrt{M+N})}\frac{|L^{\rm int} G|^2}{G} \de \sigma.
\end{equation*}
By Cauchy-Schwartz inequality, 
the right-hand side above is controlled by 
\begin{align*}
&\int_{\Sp^{M+N-1}(\sqrt{M+N})} \int_{SO(N)} \frac{|L^{\rm int} (F\circ R^{-1}(\vau, S^T \weh))|^2}{F\circ R^{-1}(\vau, S^T \weh)} \de S\, \de \sigma\\
=&\int_{\Sp^{M+N-1}(\sqrt{M+N})} \int_{SO(N)} \frac{|L (F\circ R^{-1}(\vau, S^T \weh))|^2}{F\circ R^{-1}(\vau, S^T \weh)}-\frac{|L^{\rm out} (F\circ R^{-1}(\vau, S^T \weh))|^2}{F\circ R^{-1}(\vau, S^T \weh)} \de S\, \de \sigma.
\end{align*}
Our claimed result then simply follows from the invariance of angular momentum with respect to rotations.
\end{proof}

From now on, we assume $F$ is invariant under rotations acting on $\weh$. Thus $L^{\rm out} F=0$.
However, with non-generic rotation $R^{-1}$, $|L^{\rm out} (F\circ R^{-1})|>0$.
Our next lemma gives an explicit formula for the contribution of this term.
In order to state it, let us note that by the rotation invariance of $F$ in the second variable, 
\begin{equation}\label{def:ell}
\nabla_\weh F = \ell(\vau, |\weh|) \weh 
\end{equation}
for some function $\ell$.
\begin{lemma}\label{lem:Lout}
We have
\begin{align*}
\int_{\Sp^{M+N-1}(\sqrt{M+N})} \frac{|L^{\rm out} (F\circ R^{-1})|^2}{F\circ R^{-1}} \de \sigma
=\int_{\Sp^{M+N-1}(\sqrt{M+N})}\frac{\Sigma_1}{F} \de \sigma,
\end{align*}
where
\begin{align*}
\Sigma_1:= &
|C \vau|^2 |C \nabla_\vau F|^2 - (C \vau \cdot C\nabla_\vau F)^2\nonumber \\
&+ \frac{|\weh|^2}{N}\left[  |C(\ell \vau -\nabla_\vau F)|^2{\rm Tr } (D^TD) -|D^TC(\ell \vau - \nabla_\vau F)|^2 \right]. \nonumber
\end{align*}
\end{lemma}
The proof follows from a long computation. We leave it to the appendix.

Combing Lemmas \ref{lem:ave1} with \ref{lem:Lout}, we have, after another long but direct computation, the following.
\begin{lemma}\label{lem:ave2}
We have
\begin{equation}
\int_{\Sp^{M+N-1}(\sqrt{M+N})} \frac{|LG|^2}{G} \de \sigma \le  \int_{\Sp^{M+N-1}(\sqrt{M+N})} \frac{\Sigma_2}{F} \de \sigma \nonumber
\end{equation}
where
\begin{align*} 
 \Sigma_2 =& |\vau|^2 |A\nabla_\vau F|^2 +|A\vau|^2 |\nabla_\vau F|^2  - 2(\vau\cdot \nabla_\vau F) (A \vau\cdot  A\nabla_\vau F) \nonumber \\
&+ (1-\frac{M-1}{N})|\weh|^2  |A(\ell \vau -\nabla_\vau F)|^2 + \frac{M}{N} |\weh|^2 |(\ell \vau -\nabla_\vau F)|^2  \ . \nonumber
\end{align*}
\end{lemma}
The proof of this lemma is presented in the appendix.

We come now to the proof of Theorem \ref{spherical}. As in the proof of Theorem  \ref{main1} we have that
\begin{equation} \label{totalsum}
I(\mathcal N e^{\mathcal L_S t} F_0 ) \le e^{-\Lambda t}\sum_{k=0}^\infty \frac{(M+N)^kt^k}{k!}
\sum_{\alpha_1, \dots, \alpha_k} \int_0^{2\pi}\nu(\de \theta_{\alpha_1})  \cdots \int_0^{2\pi}\nu(\de \theta_{\alpha_k}) I( \mathcal N F_0\circ R^{-1}(\underline{\alpha}, \underline{\theta}) ) \ .
\end{equation}
By Lemma \ref{lem:ave2} and Lemma \ref{sumrule} we find
\begin{align*}
&\sum_{\alpha_1, \dots, \alpha_k} \int_0^{2\pi}\nu(\de \theta_{\alpha_1})  \cdots \int_0^{2\pi}\nu(\de \theta_{\alpha_k}) I( \mathcal N F_0\circ R^{-1}(\underline{\alpha}, \underline{\theta}) ) \nonumber \\
\le & \int_{\Sp^{M+N-1}(\sqrt{M+N})} \mathcal{C}_k \frac{2|\vau|^2 |\nabla_\vau F_0|^2 -2(\vau \cdot \nabla_\vau F_0)^2 
+(1-\frac{M-1}{N})|\weh|^2  |\ell \vau -\nabla_\vau F_0)|^2}{F_0} \nonumber \\
&+\frac{M}{N}  \int_{\Sp^{M+N-1}(\sqrt{M+N})} \frac{|\weh|^2 |\ell \vau -\nabla_\vau F_0|^2}{F_0} \de \sigma \ .
\end{align*}
Inserting this estimate into the equation \eqref{totalsum} yields
\begin{align*}
&I(\mathcal N e^{\mathcal L_S t} F_0 ) \le  \left(\frac{M}{M+N} + e^{-\mu t} \frac{N}{M+N} \right) \times \nonumber \\
&\times \int_{\Sp^{M+N-1}(\sqrt{M+N})} \frac{2|\vau|^2 |\nabla_\vau F_0|^2 -2(\vau \cdot \nabla_\vau F_0)^2 
+(1-\frac{M-1}{N})|\weh|^2  |\ell \vau -\nabla_\vau F_0)|^2}{F_0} \de \sigma \nonumber \\
&+ \frac{M}{N}  \int_{\Sp^{M+N-1}(\sqrt{M+N})} \frac{|\weh|^2 |\ell \vau -\nabla_\vau F_0|^2}{F_0} \de \sigma \ .
\end{align*}
Elementary estimates yield
\begin{align*}
I(\mathcal N e^{\mathcal L_S t} F_0 ) \le  \left[ 2\left(\frac{M}{M+N} + e^{-\mu t} \frac{N}{M+N} \right) +
\frac MN\right]  \int_{\Sp^{M+N-1}(\sqrt{M+N})} \frac{|LF_0|^2}{F_0} \de \sigma,
\end{align*}
which is the claimed result. \qed

\appendix
\section{Proofs of Lemmas \ref{lem:Lout} and \ref{lem:ave2}}
\subsection{Proof of Lemma \ref{lem:Lout}}
We have
\begin{align*}
\mathcal{I}:
=&\int_{\Sp^{M+N-1}(\sqrt{M+N})} \frac{|L^{\rm out} (F\circ R^{-1})|^2}{F\circ R^{-1}} \de \sigma\\
=&\int_{\Sp^{M+N-1}(\sqrt{M+N})} \frac{\sum^{\mathrm{out}} |\omega_i( C\nabla_{\vau}F+D\nabla_{\weh}F)_j-\omega_j ( SC\nabla_{\vau}F+D\nabla_{\weh}F)_i|^2}{F}  \de \sigma,
\end{align*}
where the argument of the functions is of the form $(A^T\vau+C^T \weh, B^T\vau+D^T \weh)$, and $\sum^{\rm out}:=\sum_{M+1\leq i<j\leq N+M}$.
Next we change variable $\vau=A\vau+B\weh$, $\weh=C\vau+D\weh$, we have
\begin{align*}
\mathcal{I}&=\int_{\Sp^{M+N-1}(\sqrt{M+N})} \de \sigma\\ 
&\qquad \frac{\sum^{\mathrm{out}}
 |(C\vau+D\weh)_i( C\nabla_{\vau}F+D\nabla_{\weh}F)_j-(C\vau+D\weh)_j (C\nabla_{\vau}F+D\nabla_{\weh}F)_i|^2}{F}\\
 &=: \int_{\Sp^{M+N-1}(\sqrt{M+N})} \frac{\Sigma_0}{F} \de \sigma
\end{align*}
By \eqref{def:ell}, we have
\begin{align*}
\Sigma_0
=&\sum^{\mathrm{out}}
 |(C\vau+D\weh)_i( C\nabla_{\vau}F+\ell D \weh)_j-(C\vau+D\weh)_j (C\nabla_{\vau}F+\ell D\weh)_i|^2\\
 =&\sum^{\mathrm{out}}
 |(C\vau+D\weh)_i (C (\nabla_{\vau} F-\ell \vau))_j-(C\vau+D\weh)_j (C (\nabla_{\vau} F-\ell \vau))_i|^2.
\end{align*}
Now we average over the directions of the vector $\weh$. Recall that the function $F$ does not depend on this direction nor does $\ell$. Denote such an average by $\langle \cdot \rangle$. Hence we get, after averaging over each pair $(i,j)$,
\begin{align*}
\langle \Sigma_0\rangle=&\sum^{\mathrm{out}} |(C\vau)_i(C (\nabla_{\vau} F-\ell \vau))_j-(C\vau)_j C (\nabla_{\vau} F-\ell \vau))_i|^2\\
&+\sum^{\mathrm{out}}\Big\langle |(D\weh)_i(C (\nabla_{\vau} F-\ell \vau))_j-(D\weh)_j(C (\nabla_{\vau} F-\ell \vau))_i|^2 \Big\rangle\\
=& |C\vau|^2 |C \nabla_{\vau} F |^2-(C\vau \cdot C \nabla_{\vau} F)^2\\
&+\Big\langle |D\weh|^2 \Big\rangle |C (\nabla_{\vau} F-\ell \vau)|^2-\Big\langle (\weh, D^T C (\nabla_{\vau} F-\ell \vau))^2\Big\rangle
\end{align*}
where we used that $\langle \omega_i\rangle=0$.
Next, by using that 
$$\langle \omega_i\, \omega_j\rangle=\frac{|\weh|^2}{N} \delta_{ij},$$
we have
\begin{align*}
\langle\Sigma_0\rangle=&|C \vau|^2 |C \nabla_\vau F|^2 -(C \vau \cdot C\nabla_\vau F)^2 + \\
&\frac{|\weh|^2}{N} \left[ {\rm Tr } (D^T D) |C(\ell \vau -\nabla_\vau F)|^2 -|D^TC(\ell \vau - \nabla_\vau F)|^2 \right]
\end{align*}
Finally, just note that due to the fact that $F$ does not depend on the direction of $\weh$, we have
\begin{align*}
\mathcal{I}= \int_{\Sp^{M+N-1}(\sqrt{M+N})} \frac{\langle \Sigma_0 \rangle}{F} \de \sigma.
\end{align*}
This proves the lemma.
\qed

\subsection{Proof of Lemma \ref{lem:ave2}}
Combining Lemma \ref{lem:ave1} with \ref{lem:Lout}, it suffices to show
\begin{equation}\label{eq:lem2}
|LF|^2-\Sigma_1\leq \Sigma_2.
\end{equation}
 
Using the rotation invariance of $F$ in the second variable we write
\begin{eqnarray}
|LF|^2 &=& \sum_{i<j} |v_i\partial_{v_j} F - v_j \partial_{v_i} F|^2 + \sum_i \sum_j |w_j \partial_{v_i}F - v_i \partial_{w_j} F|^2 \nonumber. \\ 
&= &
|\vau |^2 |\nabla_\vau F|^2 - (\vau \cdot \nabla_\vau F)^2 + |\weh|^2 |\nabla_\vau F|^2 + |\vau|^2 |\weh|^2 \ell^2 - 2|\weh|^2 \ell (\vau \cdot \nabla_\vau F)  \nonumber \\
&=& |\vau |^2 |\nabla_\vau F|^2 - (\vau \cdot \nabla_\vau F)^2 + |\weh|^2 |\ell \vau - \nabla_\vau F|^2 \ . \nonumber
\end{eqnarray}
It remains, thus, to estimate  
\begin{eqnarray}\label{eq:1}
& &|\vau |^2 |\nabla_\vau F|^2 - (\vau \cdot \nabla_\vau F)^2 + |\weh|^2 |\ell \vau - \nabla_\vau F|^2 \nonumber \\
&-& |C \vau|^2 |C \nabla_\vau F|^2 + (C \vau \cdot C\nabla_\vau F)^2 \nonumber \\
&-&\frac{|\weh|^2}{N} \left[  |C(\ell \vau -\nabla_\vau F)|^2{\rm Tr } D^TD -|D^TC(\ell \vau - \nabla_\vau F)|^2 \right]
\end{eqnarray}
Since $R$ is a rotation we have that
$$
\left( \begin{array}{cc} A & B \\ C & D \end{array} \right) \left( \begin{array}{cc} A^T & C^T \\ B^T & D^T \end{array} \right) =
\left( \begin{array}{cc} I_M & 0 \\ 0 & I_N \end{array} \right)
$$
and
$$
\left( \begin{array}{cc} A^T & C^T \\ B^T & D^T \end{array} \right)\left( \begin{array}{cc} A & B \\ C & D \end{array} \right)
=\left( \begin{array}{cc} I_M & 0 \\ 0 & I_N \end{array} \right) \ .
$$
Using this we get that
$$
{\rm Tr} D^TD = N-{\rm Tr}B^TB = N-M + {\rm Tr}A^TA \ , A^TA+C^TC = I_M
$$
and
$$
D^TC = -B^TA \ .
$$ 
Using these relations \eqref{eq:1} becomes
\begin{align*}
& |\vau |^2 |\nabla_\vau F|^2 - (\vau \cdot \nabla_\vau F)^2 + |\weh|^2 |\ell \vau - \nabla_\vau F|^2  \\
-& (| \vau|^2 - |A \vau|^2)(  | \nabla_\vau F|^2-|A \nabla_\vau F|^2) +  \left[(\vau, \nabla_\vau F)  -( A \vau\cdot A\nabla_\vau F)\right]^2  \\
-&\frac{|\weh|^2}{N} \left[ (  |(\ell \vau - \nabla_\vau F)|^2-|A(\ell \vau -\nabla_\vau F)|^2)(N-{\rm Tr }B^TB) -|B^TA(\ell \vau - \nabla_\vau F)|^2 \right] 
\end{align*}
which can be simplified to
\begin{align*} 
& |\vau|^2 |A\nabla_\vau F|^2 +|A\vau|^2 |\nabla_\vau F|^2 - |A\vau|^2 |A\nabla_\vau F|^2   + (A \vau \cdot A\nabla_\vau F)^2 - 2(\vau\cdot \nabla_\vau F) (A \vau\cdot A\nabla_\vau F)  \\
+&|\weh|^2  |A(\ell \vau -\nabla_\vau F)|^2  \\
+&\frac{|\weh|^2}{N} \left[ (  |(\ell \vau -\nabla_\vau F)|^2-|A(\ell \vau -\nabla_\vau F)|^2){\rm Tr }B^TB +|B^TA(\ell \vau - \nabla_\vau F)|^2 \right] \ . 
\end{align*}
Since $ AA^T+BB^T = I_M$ and ${\rm Tr} B^TB = {\rm Tr} BB^T = M - {\rm Tr} AA^T$
we obtain an upper bound
\begin{align*} 
& |\vau|^2 |A\nabla_\vau F|^2 +|A\vau|^2 |\nabla_\vau F|^2 - |A\vau|^2 |A\nabla_\vau F|^2   + (A \vau \cdot A\nabla_\vau F)^2 - 2(\vau \cdot \nabla_\vau F) (A \vau \cdot A\nabla_\vau F)   \\
+&|\weh|^2  |A(\ell \vau -\nabla_\vau F)|^2 \\
+&\frac{|\weh|^2}{N} \left[ (  |(\ell \vau -\nabla_\vau F)|^2-|A(\ell \vau -\nabla_\vau F)|^2)M +|A(\ell \vau - \nabla_\vau F)|^2 \right] \ .
\end{align*}
which can be simplified to
\begin{align*}
& |\vau|^2 |A\nabla_\vau F|^2 +|A\vau|^2 |\nabla_\vau F|^2 - |A\vau|^2 |A\nabla_\vau F|^2   + (A \vau \cdot A\nabla_\vau F)^2 - 2(\vau \cdot \nabla_\vau F) (A \vau \cdot A\nabla_\vau F)   \\
+& (1-\frac{M-1}{N})|\weh|^2  |A(\ell \vau -\nabla_\vau F)|^2 + \frac{M}{N} |\weh|^2 |(\ell \vau -\nabla_\vau F)|^2  \ .
\end{align*}

The terms quartic in $A$ we do not know how to handle but noting that $- |A\vau|^2 |A\nabla_\vau F|^2   + (A \vau \cdot A\nabla_\vau F)^2\le 0$ we get the bound
\begin{eqnarray} 
|LF|^2-\Sigma_1\leq & &  |\vau|^2 |A\nabla_\vau F|^2 +|A\vau|^2 |\nabla_\vau F|^2  - 2(\vau \cdot \nabla_\vau F) (A \vau \cdot A\nabla_\vau F)  \nonumber \\
&+& (1-\frac{M-1}{N})|\weh|^2  |A(\ell \vau -\nabla_\vau F)|^2 + \frac{M}{N} |\weh|^2 |(\ell \vau -\nabla_\vau F)|^2=\Sigma_2  \ , \nonumber
\end{eqnarray}
as claimed.
\qed

\section{Proof of Lemma \ref{sumrule}}
One can think about $K$ as the top left entry of the matrix
\begin{align*}
\sum_{\alpha_1, \dots, \alpha_k} \lambda_{\alpha_1} \cdots 
\lambda_{\alpha_k} \int_{[-\pi,\pi]^k} \nu(\mathrm{d} \theta_1)  \cdots 
\nu(\mathrm{d} \theta_k) \, \left[\prod_{l=1}^k 
r_{\alpha_l}(\theta_l)\right]^{-1}  \begin{pmatrix} I_M & 0 \\ 0 & 
0\end{pmatrix} \left[\prod_{l=1}^k r_{\alpha_l}(\theta_l)\right] \ . 
\end{align*}
The computation hinges on a repeated application of the elementary identity 
\begin{align*}
& \int_{-\pi}^{\pi} \nu(\mathrm{d}\theta) \,  
\begin{pmatrix}\cos(\theta) & -\sin(\theta)\\ \sin(\theta)  & 
\cos(\theta)\end{pmatrix}
\begin{pmatrix} m_1 & 0\\ 0 & m_2\end{pmatrix}
\begin{pmatrix} \cos(\theta) & \sin(\theta)\\ -\sin(\theta) & 
\cos(\theta)\end{pmatrix} \\
& = \begin{pmatrix} (1-\tilde\nu) m_1+\tilde \nu m_2 & 0\\
0 & (1-\tilde\nu) m_2+\tilde \nu m_1 \end{pmatrix} \ ,
\end{align*}
where $\tilde\nu=\int \nu(\mathrm{d} \theta) \, \sin^2(\theta).$
For this to be true we just need \eqref{assumprho}.
We easily check that for the 
rotations $r_\alpha(\theta)$ 
\begin{align} \label{eq:proofsumrule}
& \sum_{\alpha} \lambda_{\alpha} \int_{-\pi}^{\pi} \nu(\mathrm{d} \theta) \, 
r_\alpha(\theta)^{-1}  \begin{pmatrix}  m_1 I_M & 0 \\ 0 &  
m_2 I_N\end{pmatrix}   r_\alpha(\theta)\notag \\ 
& = \frac{1}{\Lambda}\left(\frac{M\lambda_S}{2}+\frac{N\lambda_R}{2}\right)
\begin{pmatrix}  m_1 I_M & 0 \\ 0 &  
m_2 I_N\end{pmatrix} \notag \\
& \qquad + \frac{\mu}{\Lambda N}\begin{pmatrix}   
N(M-1)+N((1-\tilde\nu) m_1 + \tilde\nu m_2)I_M & 0 \\0 & (N-1)M+ 
 M(\tilde\nu m_1 +(1-\tilde\nu)m_2) I_N \end{pmatrix} \notag 
\\
& = \begin{pmatrix}  m_1 I_M & 0 \\ 0 &  
m_2 I_N\end{pmatrix}+ \frac{\mu_\nu}{\Lambda N}
\begin{pmatrix}  
N( m_2-m_1)I_M & 0 \\
0 &  M(m_1 -m_2)I_N
    \end{pmatrix} \ .
\end{align}
where $\mu_\nu=\tilde \nu \mu$. Denote by $L(\nu_1,\nu_2)$  the $(N+M)\times(N+M)$ matrix  
\begin{align*}
L(m_1,m_2)=\begin{pmatrix}  m_1 I_{M} & 0\\ 0 & 
m_2 I_N\end{pmatrix} \ ,
\end{align*}
and set 
\begin{align*}
\mathcal P=I_2 - \frac{\mu_\nu}{\Lambda N}\begin{pmatrix} 
                                            N & -N\\
                                            -M & M
                                           \end{pmatrix} \ .
\end{align*}
Then  \eqref{eq:proofsumrule} is recast as
\begin{align} \label{eq: proofsumrulereduced}
 \sum_{\alpha} \lambda_{\alpha} \int_{-\pi}^{\pi} \nu(\mathrm{d} \theta) \, 
r_\alpha(\theta)^{-1}L(m_1,m_2)r_\alpha(\theta)= L(m_1',m_2') \ ,
\end{align}
where 
\begin{align*}
\begin{pmatrix}  m_1' \\ m_2'\end{pmatrix}=
\mathcal P\begin{pmatrix}  m_1 \\ m_2\end{pmatrix} \ .
\end{align*}
By a repeated application of \eqref{eq: proofsumrulereduced} we obtain
\begin{align*}
& \sum_{\alpha_1, \dots, \alpha_k} \lambda_{\alpha_1} 
\cdots \lambda_{\alpha_k} \int_{[-\pi,\pi]^k} \nu(\mathrm{d}\theta_1) \, 
\cdots \nu(\mathrm{d}\theta_k) \, \left[\prod_{j=1}^k r_{\alpha_j}(\theta_j)\right]^T L(\underline m)\left[\prod_{j=1}^k r_{\alpha_j}(\theta_j)\right] = \vphantom{\sum_{\alpha_j}} L(\mathcal P^k\underline m) \ .
\end{align*}
Thus,
\[
K =\left(\mathcal P^k\, \begin{pmatrix}  1 \\ 
0\end{pmatrix}\right)_1 I_M \ .
\]
It is easy to see that $\mathcal P$ has eigenvalues $\ell_1=1$ and 
$\ell_2=1-\mu_\nu (M+N)/(\Lambda N)$ with eigenvectors $\underline 
m_1=(1, 1)$ and $\underline m_2=(N,-M)^T/(M+N)$. Consequently, 
\[
 \begin{pmatrix}  1 \\ 0\end{pmatrix}=\frac{M}{N+M}\underline 
m_1+\underline m_2 \ ,
\]
which yields
\[
 \left(\mathcal P^k\, \begin{pmatrix}  1 \\ 
0\end{pmatrix}\right)_1=\frac{M}{N+M}+ \frac{N}{M+N}\left(1-\mu_\nu
\frac{M+N}{\Lambda N}\right)^k \ .
\]
This proves Lemma \ref{sumrule}. \hfill $\qed$

\bibliography{nonequi}{}

\begin{thebibliography}{1}

\bibitem{BLTV}
F.~Bonetto, M.~Loss, H.~Tossounian, and R.~Vaidyanathan.
\newblock Uniform approximation of a {M}axwellian thermostat by finite
  reservoirs.
\newblock {\em Comm. Math. Phys.}, 351(1):311--339, 2017.

\bibitem{BGLR}
Federico Bonetto, Alissa Geisinger, Michael Loss, and Tobias Ried.
\newblock Entropy decay for the {K}ac evolution.
\newblock {\em Comm. Math. Phys.}, 363(3):847--875, 2018.

\bibitem{CCL1}
Eric Carlen, M.~C. Carvalho, and Michael Loss.
\newblock Many-body aspects of approach to equilibrium.
\newblock In {\em Journ\'ees ``\'{E}quations aux {D}\'eriv\'ees {P}artielles''
  ({L}a {C}hapelle sur {E}rdre, 2000)}, pages Exp.\ No.\ XI, 12. Univ. Nantes,
  Nantes, 2000.

\bibitem{amit}
Amit Einav.
\newblock On {V}illani's conjecture concerning entropy production for the {K}ac
  master equation.
\newblock {\em Kinet. Relat. Models}, 4(2):479--497, 2011.

\bibitem{Jeanvresse}
Elise Janvresse.
\newblock Spectral gap for {K}ac's model of {B}oltzmann equation.
\newblock {\em Ann. Probab.}, 29(1):288--304, 2001.

\bibitem{kac}
M.~Kac.
\newblock Foundations of kinetic theory.
\newblock In {\em Proceedings of the {T}hird {B}erkeley {S}ymposium on
  {M}athematical {S}tatistics and {P}robability, 1954--1955, vol. {III}}, pages
  171--197, Berkeley and Los Angeles, 1956. University of California Press.

\bibitem{KacBook}
Mark Kac.
\newblock {\em Probability and related topics in physical sciences}, volume
  1957 of {\em With special lectures by G. E. Uhlenbeck, A. R. Hibbs, and B.
  van der Pol. Lectures in Applied Mathematics. Proceedings of the Summer
  Seminar, Boulder, Colo.}
\newblock Interscience Publishers, London-New York, 1959.

\bibitem{Ledoux}
M.~Ledoux.
\newblock On an integral criterion for hypercontractivity of diffusion
  semigroups and extremal functions.
\newblock {\em J. Funct. Anal.}, 105(2):444--465, 1992.

\bibitem{villani}
C{\'e}dric Villani.
\newblock Cercignani's conjecture is sometimes true and always almost true.
\newblock {\em Comm. Math. Phys.}, 234(3):455--490, 2003.

\end{thebibliography}
\bibliographystyle{plain}

\end{document}